\documentclass[11pt, runningheads]{llncs}
\usepackage[T1]{fontenc}
\usepackage{notoccite}
\usepackage{float}

\usepackage{mathtools}
\usepackage{multirow}
\makeatletter
\newcases{centercases}{\quad}
  {\hfil$\m@th\displaystyle{##}$\hfil}
  {$\m@th\displaystyle{##}$\hfil}{\lbrace}{.}
\makeatother
\usepackage{mathtools} 
\usepackage{amsmath}
\usepackage{float}
\usepackage{braket}
\usepackage{amssymb}
\usepackage{algorithm}
\usepackage{caption}
\usepackage{algpseudocode}
\usepackage{mathtools}

\usepackage{graphicx}
\usepackage[margin=1in,footskip=0.25in]{geometry}
\begin{document}
\title{Efficient algorithm for linear diophantine equations in two variables}
%
%
\author{Mayank Deora \and
Pinakpani Pal}
\authorrunning{M. Deora, P. Pal}
%
\institute{Electronics and Communication Sciences Unit, Indian Statistical Institute, Kolkata - 700108, India
\\
\email{mbdeora@gmail.com, pinak@isical.ac.in}\\
}
\maketitle              
\begin{abstract}
Solving linear diophantine equations in two variables have applications in computer science and mathematics. In this paper, we revisit an algorithm for solving linear diophantine equations in two variables, which we refer as DEA-R algorithm. The DEA-R algorithm always incurs equal or less number of recursions or recursive calls as compared to extended euclidean algorithm. With the objective of taking advantage of the less number of recursive calls , we propose an optimized version of the DEA-R algorithm as DEA-OPTD. In the recursive function calls in DEA-OPTD, we propose a sequence of more efficient computations. We do a theoretical comparison of the execution times of DEA-OPTD algorithm and DEA-R algorithm to find any possible bound on the value of $c$ for DEA-OPTD being better than DEA-R. We implement and compare an iterative version of DEA-OPTD (DEA-OPTDI) with two versions of a widely used algorithm on an specific input setting. In this comparison, we find out that our algorithm outperforms on the other algorithm against atleast 96\% of the inputs.

\end{abstract}

\section{Introduction}
For solving linear diophantine equations in two variables, we need to find integer values of the variables. These   linear diophantine  equations are written as:
\begin{equation}
   ax+by=c \label{eq1} 
\end{equation}
   
where $a$, $b$ and $c$ are known integers. $x$ and $y$ are integer variables. Due to bazout's lemma \cite[p.~10]{seroul2000programming}, their solution will exist only if either $c=0$ or it is divisible by greatest common divisor of $a$ and $b$.
 \par In \cite{deora2023algorithm}, an algorithm for solving two variable linear diophantine equations have been proposed, whose worst case time complexity is $\theta(\log n)$ ($\log n$ is  $max(a,b)$). Extended euclid's algorithm is the most widely used algorithm for solving linear diophantine equations in two variables\cite{knuth2014art}. It is an extension of euclidean algorithm for computing gcd (greatest common divisor). Similarly binary gcd algorithm can also be extended to solve two variable linear diophantine equations \cite{knuth2014art}. The time complexity of both extended euclid's algorithm and extension of binary gcd algorithm is $\theta(\log n)$ ($n$ is the bit-size of $max(a,b)$).
     In \cite{lankford1989non} and \cite{ajili1995complete} there are  algorithms which are used for finding non-negative integer basis of linear diophantine equations.   
     \cite{clausen1989efficient} introduces a graph based algorithm, which can find minimal solutions of a diophantine equation of the type $ax-by=c$. There, the time complexity of the graph based algorithm is not discussed but its implementation given in the same paper clarifies that its time requirement is proportional to $n$, where $n$ is the maximum value of the coefficients ($\max(a,b)$). 
     \cite{clausen1989efficient} and \cite{lankford1989non} propose algorithms for solving linear diophantine equations, which take exponential time proportional to the bit-size of the input,  but they are applicable in associative-commutative unification. \cite{kumar2022alternative} and \cite{sumarti2023method} propose algorithms, which are based on optimization techniques to solve diophantine equations.  As a part of the implementation of these algorithms, some instances of two variable linear diophantine equations are solved. 
      
     \cite{bradley1971algorithms}, \cite{chou1982algorithms}, \cite{lazebnik1996systems} and \cite{esmaeili2001class} propose algorithms for solving system of linear diophantine equations. There are algorithms to solve  system of linear diophantine equations with upper bound and lower bound on variables in \cite{ramachandran2006use} and \cite{aardal2000solving}. 
     \par We propose an algorithm to solve linear diophantine equations in two variables which is an optimized version of the algorithm proposed in \cite{deora2023algorithm}, referred as DEA-R algorithm.  There, in each recursion, the DEA-R algorithm uses the value of $c$ and previous value returned by the  recursive function to compute solutions. The repetitive usage of $c$ in each recursion implies that the bit size of $c$ is a decisive argument in finding the time complexity of the algorithm. 
     \par In the current work, we modify DEA-R such that the value of $c$ is partially removed from each recursion, so that it is not dominant in determining the time complexity of the algorithm. The proposed algorithm is referred as DEA-OPTD and  original version is referred as DEA-R. Both algorithms incur equal or less number of recursions or recursive calls as compared to extended euclidean algorithm. Less number of recursive calls does not imply less execution time, because in each recursive call of DEA-R (or DEA-OPTD) and extended euclid's algorithm, there are different arithmetic operations. In DEA-OPTD, we take advantage from the less number of recursive calls by reducing the size of operands in some arithmetic operations in the recursive function. 
     \par We perform theoretical comparison between DEA-R and DEA-OPTD algorithm to find some bound (or inequality )containing  the values of $c$ and intermediate values observed during subsequent recursive calls. When this bound is satisfied we conclude that DEA-OPTD algorithm is better than DEA-R algorithm.  

     \par Since DEA-OPTD algorithm is recursive, it must have a run time overhead in implementation. So, we propose an iterative version of DEA-OPTD algorithm referred as DEA-OPTDI.
     As a case study, we implement DEA-OPTDI algorithm and compare it with DEA-I algorithm (given in \cite{deora2024averagecaseefficientalgorithm}) and two variations of extended euclid's algorithm. DEA-I algorithm is iterative version of DEA-R algorithm. DEA-I algorithm is given in \cite{deora2024averagecaseefficientalgorithm}. We refer the two versions of extended euclid's algorithm as EEA-I and EEA-2. EEA-I (given in \cite{deora2024averagecaseefficientalgorithm}) is an iterative version of extended euclid's algorithm. EEA-2 algorithm is given on the github link \cite{EEAGit1}. For the comparison, we select 100000 random input instances $(a,b,c)$ where $a,b$ and $c$, all the three parameters lie between 1 to $2^{10}$ (both inclusive). After the implementation of this  comparison, we find out that we are able to optimize DEA-R in the form of DEA-OPTD on an specific input setting. The comparison result depicts that on the chosen set of inputs DEA-OPTDI is always the most efficient algorithm among the four algorithms.
      
\section{Notations} 
We use function call, recursive function call and recursive call interchangeably throughout this paper. All the terms refer the same thing here, because the function(in DEA-R or in DEA-OPTD algorithm), that we refer in theoretical analysis is a recursive function.
We assume that input in the first recursive function call of the proposed algorithm, is ($a_1,a_2,c$). We assume that quotient of division of $a_1$ by $a_{2}$ is $q_1$ and the remainder (if not zero) is $a_{3}$. 
Similarly for any $i>1$, quotient of division of $a_i$ by $a_{i+1}$ is $q_i$ and remainder will be $a_{i+2}$. 
i.e.
$$
a_i=q_ia_{i+1}+a_{i+2}
$$
Here, $a_{i+2}$ can not be 0.
In the proposed algorithm, after first recursive function call, inputs which are passed to the  next recursive function calls are $(a_2,a_3,c)$, $(a_3,a_4,c)..........$ upto $(a_{k},a_{k+1},c)$. We reserve the notations  $\mathbf{c_1},\mathbf{c_2},\mathbf{c_3},.........,\mathbf{c_k}$ to represent the sets of special values of $c$. Each $c \in \mathbf{c_i}$ is  in the following general form:

$$c=a_i+Q_ia_{i+1}$$ where $Q_i$ can be any integer (positive, negative or 0). 

\par Since in each recursive call in DEA-R and DEA-OPTD algorithms value of $c$ does not change, in section 3 and appendix, we omit the third argument of function $f(a,b,c)$ and write it as $f(a,b)$. Thus $f(a,b)$ and $f(a,b,c)$ in these algorithms refer to the same recurrence. 
 \section{Algorithms}

In this section we observe the algorithm given in \cite{deora2023algorithm}, which we refer as DEA-R algorithm. 
We write a brief overview of this algorithm which is useful to propose the modified algorithm. After this we propose the 
modified version of the DEA-R algorithm, which we refer as DEA-OPTD algorithm.  

\subsection{DEA-R Algorithm}
DEA-R algorithm solves the diophantine equation $ax+by=c$ and finds integer $y$. From $y$, we compute $x$ as $x=\frac{c-by}{a}$. For an input $(a,b,c)$, let $c$ belongs to all the sets in the set $\{\mathbf{c_{i_1}},.....,\mathbf{c_{i_m}}\}$, such that $i_1<.......<i_m$. Then, on $i_1^{th}$ recursive call of the function  $f$ of DEA-R algorithm, the algorithm executes its first return statement. The reason for this is given in \cite{deora2024averagecaseefficientalgorithm}. To briefly explain the reason, we can observe that in the $i_1^{th}$ recursive call, the input to the recursive function $f$ is $(a_{i_1},a_{{i_1}+1},c)$. $f(a_{i_1},a_{{i_1}+1},c)$ executes the else if statement at line 5, which in turn computes $(c-a_{i_1}) \bmod a_{i_1+1}$. Since $c \in \mathbf{c_{i_1}}$ and c does not belong to any set  other than elements of the family of sets $\{\mathbf{c_{i_1}},.....,\mathbf{c_{i_m}}\}$, so $(c-a_{i_1}) \bmod a_{i_1+1}=0$ (satisified in $i^{th}_1$ recursive call). Subsequently the algorithm executes its first return statement at line 7 in $i_{1}^{th}$ recursive call.

\par We use this information on the set to which $c$ belongs, to propose DEA-OPTD algorithm. 
 \begin{algorithm}
\caption{\textbf{DEA-R} Algorithm \cite{deora2023algorithm}}\label{euclid}
\hspace*{\algorithmicindent} \textbf{Input}: $a,b,c$ ($a>b ,a\ne 0,b\ne 0,c\ne 0$)  \\ \hspace*{\algorithmicindent} $a,b,c$ are integers
\newline\hspace*{\algorithmicindent}\textbf{Output}: $y$ 
\newline\hspace*{\algorithmicindent} $y$ is an integer.
\begin{algorithmic}[1]

\Procedure{{\textit{$f$}}}{$a,b,c$}
\If{($b=0$)} 
\State PRINT ("$a$ is the gcd of original inputs and it does not divide $c$")
\State \textbf{exit}

\ElsIf {$((c-a) \bmod b=0)$} 
\State  
$y\gets (c-a)/ b$
\State \Return $y$
\Else 
\State 
$y\gets\left(c-f(b,a \bmod b,c)\times a)\right/b$
\State \Return $y$
\EndIf

\EndProcedure

\end{algorithmic}
\end{algorithm}

\subsection{DEA-OPTD: an optimized version of the DEA-R Algorithm}
To design DEA-OPTD, first we derive a recursive function of second order. Assume that $c\in \mathbf{c_i}$ such that $c=a_i+Q_ia_{i+1}$, and if $c$ belongs to any $c_j$ other than $c_i$, then $j$ would be greater than $i$. This implies that on intermediate inputs $(a_i,a_{i+1},c)$ (given to a recursive function call as arguments), function $f$ in DEA-R executes its first return statement. Then in the  DEA-R algorithm, the sequence of computations in the form of return values of the function $f$ is given as follows:
\begin{equation} \label{Eq_seq1}
    \begin{split}
&f(a_1,a_2)=\frac{c-f(a_2,a_3)a_1}{a_2}
\\
&\vdots
\\ &f(a_{i-1},a_{i})=
\frac{c-f(a_i,a_{i+1})a_{i-1}}{a_i}
\\&f(a_{i},a_{i+1})=
\frac{c-a_{i}}{a_{i+1}}=Q_i
\end{split}
\end{equation}
We substitute back the value of $f(a_i,a_{i+1})$  in $f(a_{i-1},a_i)$, to produce:

$$f(a_{i-1},a_i)=\frac{c-f(a_i,a_{i+1})a_{i-1}}{a_i}=\frac{c-Q_ia_{i-1}}{a_i}$$
Since $a_{i-1}=q_{i-1}a_i+a_{i+1}$ and $c=a_i+Q_ia_{i+1}$,
$$f(a_{i-1},a_i)=\frac{a_i+Q_ia_{i+1}-Q_ia_{i-1}}{a_i}=\frac{a_i+Q_ia_{i+1}-Q_iq_{i-1}a_i-Q_ia_{i+1}}{a_i}$$

$$=1-Q_iq_{i-1}$$
In lemma \ref{lem1}, we propose a recurrence relation using which we can find the value of $f(a_{i-2},a_{i-1})$ from the values of $f(a_i,a_{i+1})$ and $f(a_{i-1},a_i)$.  
\begin{lemma} \label{lem1}
$f$ is the recursive function given in DEA-R algorithm. $f(a_j,a_{j+1})$ is the value returned by the function $f(a_j,a_{j+1},c)$. If $c$ is in the set  $\mathbf{c_i}$, then for all $j<i-1$, the following recurrence holds:
\begin{equation}\label{eq3} f(a_j,a_{j+1})=f(a_{j+2},a_{j+3})-f(a_{j+1},a_{j+2})q_j\end{equation}
\end{lemma}
\begin{proof}
    $$f(a_j,a_{j+1})=\frac{c-f(a_{j+1},a_{j+2})a_j}{a_{j+1}}$$
$$=\frac{c-f(a_{j+1},a_{j+2})(q_ja_{j+1}+a_{j+2})}{a_{j+1}}$$

$$\implies f(a_j,a_{j+1})=\frac{c-f(a_{j+1},a_{j+2})a_{j+2}}{a_{j+1}}-f(a_{j+1},a_{j+2})q_j$$

    We know that since:
    $$f(a_{j+1},a_{j+2})=\frac{c-f(a_{j+2},a_{j+3})a_{j+1}}{a_{j+2}}$$
    , so 
    $$  f(a_j,a_{j+1})=\frac{c-c+f(a_{j+2},a_{j+3})a_{j+1}}{a_{j+1}}-f(a_{j+1},a_{j+2})q_j$$
    Therefore
    $$f(a_j,a_{j+1})=f(a_{j+2},a_{j+3})-f(a_{j+1},a_{j+2})q_j$$
    \qed
\end{proof}
If $c \in \mathbf{c_i}$, by the sequence of computations shown in set of equations \ref{Eq_seq1} and lemma \ref{lem1}, we have the following recurrence relation:
\begin{equation}
   \label{receq1}
f(a_n,a_{n+1})=
\begin{centercases}
  f(a_{n+2},a_{n+3})-\left \lfloor \frac{a_n}{a_{n+1}} \right 
 \rfloor f(a_{n+1},a_{n+2})      & n<i-1 \\
 1-Q_i\left \lfloor \frac{a_n}{a_{n+1}} \right \rfloor 
 &n=i-1\\
 Q_i & n=i \end{centercases} 
\end{equation}

\par We use dynamic programming with memoization  and lemma \ref{lem1} to design DEA-OPTD algorithm. In DEA-OPTD, firstly we find the arguments, $a$, $b$ and $c$ for the last recursive call to $f$. Let these last arguments be $a_i$, $a_{i+1}$ and $c$. Then by sequence of computations given by set of equations \ref{Eq_seq1}, the value returned by last function call i.e. $f(a_i,a_{i+1},c)$ will be $\frac{c-a_i}{a_{i+1}}$.  Then we find the value of second last function call and compute the value returned by subsequent calls with the help of recursive relation \ref{receq1}. 
\par The last recursive call returns $(1,\frac{c-a_i}{a_{i+1}})$ as variables $f_1,f_2$. Thus the second last return value (from line 9) will be $(f_2,f_1-f_2 \times \left \lfloor \frac{a_{i-1}}{a_i} \right \rfloor) =(\frac{c-a_i}{a_{i+1}},1-\frac{c-a_i}{a_{i+1}}\left \lfloor \frac{a_{i-1}}{a_i} \right \rfloor)$. Similarly from recursive relation \ref{receq1} which is implicit in line 9 of the DEA-OPTD algorithm, it follows that  the algorithm finally returns the value of $(f(a_2,a_3),f(a_1,a_2))$ or $(f(a_2,a_3,c),f(a_1,a_2,c))$. $a_1$ and $a_2$ are $a$ and $b$ in the initial inputs ($a,b,c$) to the algorithm. From the initial equation in the sequence of equations given by \ref{Eq_seq1},  $f(a_2,a_3)$ is the value of $x$ and $f(a_1,a_2)$ is the value of $y$.

\begin{algorithm}
\caption{\textbf{DEA-OPTD} Algorithm}\label{euclid}
\hspace*{\algorithmicindent} \textbf{Input}: Integers $a,b,c$ ($a>b ,a\ne 0,b\ne 0,c\ne 0$)  \\
\hspace*{\algorithmicindent}\textbf{Output}: Integer pair $(x,y)$ such that $ax+by=c$
\newline\hspace*{\algorithmicindent} 
\begin{algorithmic}[1]

\Procedure{{\textit{$f$}}}{$a,b,c$}
\If{$b=0$} 
\State PRINT (''$a$ is the gcd of original inputs and it does not divide $c$'')
\State \textbf{exit}

\ElsIf {$(c-a) \bmod b=0$} 

\State  \Return 
$(1,(c-a)/ b)$

\Else 
\State $f_1,f_2\gets f(b,a \bmod b,c)$
\State \Return $(f_2, f_1-f_2 \times \lfloor \frac{a}{b}\rfloor)$
\EndIf

\EndProcedure

\end{algorithmic}
\end{algorithm}

\subsection{DEA-OPTDI Algorithm} 

The recursive algorihtms have an inherent runtime overhead in function calling. So, for the implementation, we develop an iterative version of DEA-OPTD algorithm. In this paper, we refer this iterative version of the DEA-OPTD algorithm as DEA-OPTDI algorithm. The exit condition for the while loop at line 4 in DEA-OPTDI algorithm is $(c-a) \bmod b = 0$, which is same as second base condition in DEA-OPTD. To store the values of $\left\lfloor \frac{a_i}{a_{i+1}} \right\rfloor$ ($a_i$ are the integers observed during execution of DEA-R or DEA-OPTD algorithm), we use an array named as $floorarray$. No solution of a diophantine equation is identified by $b=0$ condition at line 10. If the diopantine equation $ax+by=c$ is solvable, $f_1$ gets assigned by 1 at line 17. $f_2$ also gets assigned by the value $Q_i$, where $c=a_i+Q_ia_{i+1}$. These assignments are equivalent to the last return value from line 9 in DEA-OPTD algorithm.     The computations inside the  while loop at line 19 are equivalent to computations done by recursive relation \ref{receq1}, which is also done by return statement on line 9 in DEA-OPTD.
\begin{algorithm} [H]
\caption{\textbf{DEA-OPTDI} Algorithm (Iterative version of DEA-OPTD)}\label{euclid}
\hspace*{\algorithmicindent} \textbf{Input}: $a,b,c$ ($a>b ,a\ne 0,b\ne 0,c\ne 0$)  \\ \hspace*{\algorithmicindent} $a,b,c$ are integers
\newline\hspace*{\algorithmicindent}\textbf{Output}: $(x,y)$ 
\newline\hspace*{\algorithmicindent} $y$ is an integer.
\begin{algorithmic}[1]

\Procedure{\textbf{\textit{$f$}}}{$a,b,c$}
\State $no\_solution\gets 0$
\State $arraysize\gets 0$
\While{($(c-a) \bmod b \ne 0)$}
\State $floorarray[arraysize]\gets \left \lfloor \frac{a}{b} \right \rfloor$
\State $temp \gets a$
\State $a\gets b$
\State $b \gets temp \bmod b$
\State $arraysize\gets arraysize+1$
\If{($b=0$)}
\State $no\_solution\gets 1$
\State PRINT ("$a$ is the gcd of original inputs. and it does not divide $c$") 
\State \textbf{Break}
\EndIf
\EndWhile
\If{($no\_solution=0$)}
\State $f_1 \gets 1$
\State $f_2 \gets \frac{c-a}{b}$
\While{$(arraysize \ge 1)$}
\State $temp \gets f_2$
\State $f_2=f_1-f_2*floorArray[arraysize-1]$
\State $f_1=temp$
\State $arraysize \gets arraysize-1$ 

\EndWhile
\EndIf
\State \Return $(f_1,f_2)$

\EndProcedure

\end{algorithmic}
\end{algorithm}

\subsection{Comparison between DEA-R and DEA-OPTD} In this section, first we write the motivation behind modifying the DEA-R algorithm  and then present result from the comparative theoretical analysis of the running times of the modified algorithm DEA-OPTD with the DEA-R algorithm. The motivation of modifying the DEA-R algorithm is to reduce its execution time. Note that in the modified version, there is no change in the number of recursive calls. So, for reducing the execution time, it is necessary to either minimize the size of operands used in basic arithmetic operations or  minimize the number of arithmetic operations. In DEA-OPTD, we are able to reduce the size of operands. This reduction in size of operands makes it a better candidate than the DEA-R to compare against extended euclid's algorithm. So, we do this comparison in section 4. 

\par For a theoretical comparison between the two algorithms (DEA-R and DEA-OPTD) we assume that inputs to both, DEA-R and DEA-OPTD algorithms are same. Under the same inputs, both the algorithms compute the same function and incur equal number of recursive function calls. Since the function $f$ in DEA-R algorithm and $f$ in DEA-OPTD algorithm start with the same arguments and proceed with same arguments in each recursive call, both the functions incur same number of recursive function calls. We compare the execution times of DEA-OPTD algorithm and DEA-R algorithm. Specifically, we present a theoretical comparison between the cost of arithmetic operations in the corresponding recursive calls of both DEA-R and DEA-OPTD algorithms. Recursive function $f$ in both the algorithms have the same base conditions given as follows:
\begin{enumerate}
    \item $b = 0$
    \item $(c-a) \bmod b = 0$
    
\end{enumerate}We can also see that the function $f$ in both the algorithms have three divisions, but on different arguments in all function calls, except last call. In last recursive function call, there are two divisions, if the equation $ax+by=c$ is solvable. If the equation is not solvable then both the algorithms return after satisfying the first base condition ($b = 0$). There is one multiplication in both the algorithms.

Assume that we have a recursion variable $j$ corresponding to $j^{th}$ recursion.
For the recursion variable $j$, each recursion of the function $f$ computes $f(a_j,a_{j+1})$ according to the recurrence \ref{receq1}. Since recurrence relation \ref{receq1} is derived from function $f$ of DEA-R algorithm, it is satisfied by function $f$ of both the algorithms. We make an assumption that the cost of modified algorithm DEA-OPTD is less than the original algorithm DEA-R. Under this assumption, we try to find any lower or upper bound on the value of $c$. We compare the total cost for computing $f(a_j,a_{j+1})$ (for all $j$) in both the algorithms (see appendix 6.1). By the theoretical comparison given in appendix 6.1, we find that increasing or decreasing the value of $c$ does not imply that the modified algorithm, DEA-OPTD will take less time compared to original algorithm DEA-R. Thus it requires further investigation on the value of function $f$ and $c$ to conclude some result from the comparative analysis of DEA-OPTD and DEA-R algorithm.

\section{Implementation}
The implementation is  done on a 64 bit computer with intel core i5 2.60 GHz processor. To develop the codes, we use C programming language. To handle large integers we use GMP library \cite{GMP1} in C. To count the number of cpu cycles (cpu cycle count) in execution of a C program, we use \_\_rdtsc() function in it.  
\par
 We setup an implementation task as a case study. Extended euclidean algorithm is the most widely used algorithm for our problem. So, in the implementation, we compare DEA-OPTDI with EEA-I \cite{deora2024averagecaseefficientalgorithm} and EEA-2 and compare DEA-I  algorithm \cite{deora2024averagecaseefficientalgorithm} with EEA-I and EEA-2 algorithms. EEA-I algorithm is an iterative version of extended euclid's algorithm given in. EEA-2 is also an iterative version of extended euclid's algorithm given on the github link \cite{EEAGit1}. 

\par In the implementation, $a,b$ and $c$ are chosen uniformly at random between 1 to $2^{10}$ (both inclusive). Thus each input instance ($a,b,c$) contains all the three arguments $a$,$b$ and $c$ as uniform random integers between 1 to $2^{10}$ (both inclusive). Then DEA-OPTDI, EEA-I, EEA-2 and DEA-I algorithm is executed on 100000 such input instances. This program execution is repeated 10 times. On each input, we compute the cpu cycle count for execution of all the four  algorithms. We find the percentage of inputs on which DEA-I algorithm is better than EEA-I and EEA-2. Similarly, we find the percentage of inputs on which DEA-OPTDI is better than 
EEA-I and EEA-2. One algorithm being better than the other implies that the cpu cycle count in execution of the first is less than that in the execution of the second algorithm. 
\par We present the results of implementation in table \ref{tab:1} and \ref{tab:4}. Table \ref{tab:1} depicts average  cpu cycle count in the execution of DEA-OPTDI, DEA-I, EEA-2 and EEA-I algorithm. From the table \ref{tab:1}, it is clear that in each of the 10 program executions, the cpu cycle count in execution of DEA-OPTDI is least among all the algorithms. In each of the 10 experiments, average cpu cycle count in the execution of EEA-I is the second least. \par We develop DEA-OPTDI with the aim of optimizing DEA-I algorithm. So in the impelmentation, we find out the percentage of inputs (out of 100000 inputs), on which DEA-OPTDI is better than EEA-I and EEA-2 and do the same for DEA-I. We find this result on the same inputs which were used to prepare table \ref{tab:1}. We present these results in table \ref{tab:4}, which depicts that DEA-OPTDI is better on atleast 96 \% of inputs against EEA-I whereas DEA-I is better than EEA-I on atleast 26\% of inputs, but not on more than 32 \% of inputs.  Table \ref{tab:4} depicts that DEA-OPTDI is better than EEA-2 algorithm on atleast 99 \% of inputs whereas DEA-I is better than EEA-2 on atleast 54 \% of inputs but not on more than 72\% of inputs. Thus, DEA-OPTDI is better than DEA-I. Based on this result we can say that on random inputs of size less than or equal to 10 bits, we improve the DEA-I algorithm. 

\begin{table}[H]
    \centering
    \caption{Average CPU-cycle count in execution of the four algorithms}
    \label{tab:1}\begin{tabular}{p{2.9cm}|p{2.9cm}|p{2.9cm}|p{2.9cm}}
    \hline
    \multicolumn{4}{c}{\multirow{2}{*}{\begin{tabular}[c]{@{}c@{}}Comparison of algorithms on a 64 bit computer with intel core i5 \\2.60 GHz processor using C program and GMP library\end{tabular}}} \\
        \multicolumn{4}{c}{} \\
        
    \hline
          Average cpu cycle count in execution of DEA-OPTDI & Average cpu cycle count in execution of DEA-I & Average cpu cycle count in execution of EEA-2 & Average cpu cycles count in execution of EEA-I\\
         \hline
         
         4384.7 & 7877.5 & 7304.2 & 6129.7
         \\
         \hline
         
          4392.3 & 7061.3 & 7206.4 & 6028.2
         \\
         \hline

          4683.0 & 7650.7 & 7709.4 & 6559.2
         \\
         \hline
        4545.0  & 7296.0  & 7499.1  & 6232.7  \\ \hline
6431.0  & 10435.5  & 10153.5  & 8805.3  \\ \hline
4399.1  & 7422.3  & 7376.8  & 6099.0  \\ \hline
4506.7  & 7542.4  & 7529.5  & 6294.4  \\ \hline
4272.3  & 6727.0  & 6769.5  & 5648.2  \\ \hline
4199.4  & 6546.5  & 6754.8  & 5527.8  \\ \hline
3928.2  & 5850.2  & 6501.3  & 5203.6  \\ \hline
    \end{tabular}
    
\end{table}

\begin{table}[H]
     \centering
     \caption{Comparison of DEA-OPTDI and DEA-I with EEA-2 and EEA-I on same input parameters (performed 10 times) with 100000 random inputs, each time. In each input instance ($a,b,c$), $a$, $b$ and $c$ are chosen randomly between 1 to $2^{10}$ (both inclusive).}
     \label{tab:4}\begin{tabular}{p{2.9cm} | p{2.9cm} | p{2.9cm} | p{2.9cm}} 
     \hline 
     \multicolumn{4}{c}{\multirow{2}{*}{\begin{tabular}[c]{@{}c@{}}Comparison of algorithms on a 64 bit computer with intel core i5 \\2.60 GHz processor using C program and GMP library\end{tabular}}} \\
   
        \multicolumn{4}{c}{} \\
        \hline
          \multicolumn{4}{c}{\multirow{2}{*}
     {\begin{tabular}{p{5.3cm}|p{5.3cm}} 
    percentage of inputs on which DEA-OPTDI algorithm is better than A & percentage of inputs on which DEA-I algorithm is better than A 
 \end{tabular}}} \\
  \multicolumn{2}{c}{} \\
 
        \hline
         A=EEA-I & A=EEA-2 & A=EEA-I & A=EEA-2 \\
          \hline
       
        97.6  & 99.6  & 29.6  & 69.2  \\ \hline
97.2  & 99.5  & 30.1  & 69.8  \\ \hline
97.5  & 99.6  & 31.3  & 70.9  \\ \hline
97.3  & 99.6  & 30.7  & 71.8  \\ \hline
96.4  & 99.5  & 29.1  & 54.4  \\ \hline
96.8  & 99.6  & 26.6  & 60.4  \\ \hline
97.8  & 99.6  & 28.7  & 63.1  \\ \hline
97.0  & 99.7  & 27.8  & 62.7  \\ \hline
96.9  & 99.7  & 28.8  & 68.9  \\ \hline
98.4  & 99.7  & 31.1  & 76.9  \\ \hline
     \end{tabular}
     
 \end{table}

\section{Conclusion}
The number of recursions calls incurred in DEA-R algorithm is equal or less compared to extended euclidean algorithm. It does not imply that DEA-R takes equal or less execution time than extended euclidean algorithm, when executed on a computer. We propose a recursive function 
 of second order and use it to propose the  DEA-OPTD algorithm. We are able to compare the exact execution times taken by the two algorithms, because the inputs in each recursive call and number of recursions is same. In this theoretical analysis we can not find any dependency of better performance of DEA-OPTD algorithm on the value of $a
 _i$'s or $c$. Further investigation on inequality \ref{cond15} (in appendix 6.1) is required to conclude anything about dependency of better performance of DEA-OPTD algorithm on the values of $c$ or $a_i$'s. 
 \par In the implementation of DEA-OPTD algorithm, we compare it with DEA-R and two variants of extended euclid's algorithm on an specific input setting, where values of $a,b$ and $c$ are small ($\le2^{10}$). Although in theoretical analysis, we could not conclude any relation, in implementation we verify that
for the small values of input parameters, DEA-OPTD is the most efficient algorithm.  In future, we can work to reduce a large input to smaller input, so that we can take advantage of DEA-OPTD on all kinds of inputs.

\bibliographystyle{splncs04}
 \bibliography{mybib}

\begin{thebibliography}{10}

\bibitem{seroul2000programming}
Raymond Seroul.
\newblock {\em Programming for mathematicians}.
\newblock Springer Science \& Business Media, 2000.

\bibitem{deora2023algorithm}
Mayank Deora.
\newblock An algorithm for solving two variable linear diophantine equations.
\newblock In {\em International Conference on Advances in Computing and Data Sciences}, pages 35--48. Springer, 2023.

\bibitem{knuth2014art}
Donald~E Knuth.
\newblock {\em The Art of Computer Programming: Seminumerical Algorithms, Volume 2}.
\newblock Addison-Wesley Professional, 2014.

\bibitem{lankford1989non}
Dallas Lankford.
\newblock Non-negative integer basis algorithms for linear equations with integer coefficients.
\newblock {\em Journal of Automated Reasoning}, 5:25--35, 1989.

\bibitem{ajili1995complete}
Farid Ajili and Evelyne Contejean.
\newblock Complete solving of linear diophantine equations and inequations without adding variables.
\newblock In {\em International Conference on Principles and Practice of Constraint Programming}, pages 1--17. Springer, 1995.

\bibitem{clausen1989efficient}
Michael Clausen and Albrecht Fortenbacher.
\newblock Efficient solution of linear diophantine equations.
\newblock {\em Journal of Symbolic Computation}, 8(1-2):201--216, 1989.

\bibitem{kumar2022alternative}
Nirmal Kumar.
\newblock An alternative computational optimization technique to solve linear and nonlinear diophantine equations using discrete wqpso algorithm.
\newblock {\em Soft Computing}, 26(22):12531--12544, 2022.

\bibitem{sumarti2023method}
Novriana Sumarti, Kuntjoro~Adji Sidarto, Adhe Kania, Tiara~Shofi Edriani, and Yudi Aditya.
\newblock A method for finding numerical solutions to diophantine equations using spiral optimization algorithm with clustering (soac).
\newblock {\em Applied Soft Computing}, page 110569, 2023.

\bibitem{bradley1971algorithms}
Gordon~H Bradley.
\newblock Algorithms for hermite and smith normal matrices and linear diophantine equations.
\newblock {\em Mathematics of Computation}, 25(116):897--907, 1971.

\bibitem{chou1982algorithms}
Tsu-Wu~J Chou and George~E Collins.
\newblock Algorithms for the solution of systems of linear diophantine equations.
\newblock {\em SIAM Journal on computing}, 11(4):687--708, 1982.

\bibitem{lazebnik1996systems}
Felix Lazebnik.
\newblock On systems of linear diophantine equations.
\newblock {\em Mathematics Magazine}, 69(4):261--266, 1996.

\bibitem{esmaeili2001class}
Hamid Esmaeili, Nezam Mahdavi-Amiri, and Emilio Spedicato.
\newblock A class of abs algorithms for diophantine linear systems.
\newblock {\em Numerische Mathematik}, 90:101--115, 2001.

\bibitem{ramachandran2006use}
Parthasarathy Ramachandran.
\newblock Use of extended euclidean algorithm in solving a system of linear diophantine equations with bounded variables.
\newblock In {\em Algorithmic Number Theory: 7th International Symposium, ANTS-VII, Berlin, Germany, July 23-28, 2006. Proceedings 7}, pages 182--192. Springer, 2006.

\bibitem{aardal2000solving}
Karen Aardal, Cor~AJ Hurkens, and Arjen~K Lenstra.
\newblock Solving a system of linear diophantine equations with lower and upper bounds on the variables.
\newblock {\em Mathematics of operations research}, 25(3):427--442, 2000.

\bibitem{deora2024averagecaseefficientalgorithm}
Mayank Deora and Pinakpani Pal.
\newblock An average case efficient algorithm for solving two variable linear diophantine equations, 2024.

\bibitem{EEAGit1}
{Extended Euclidean Algorithm} kernel description.
\newblock \url{https://github.com/blichtleAlt/ExtendedEuclideanAlgorithm\\/blob/master/euclid.cpp}.
\newblock Accessed: 2024-09-24.

\bibitem{GMP1}
{The GNU MP Bignum library} kernel description.
\newblock \url{http://https://gmplib.org/}.
\newblock Accessed: 2024-09-24.

\end{thebibliography}

\section{Appendix}

\subsection{Comparison between DEA-R and DEA-OPTD Algorithm}

We assume that all the assignment operations take negligible time, and arithmetic operations take worst case time (i.e. with no optimization).Since we want to compare those instructions only which are different in both the algorithms, we compare line 9 of DEA-R algorithm with line 9 of DEA-OPTD algorithm. We assume that DEA-OPTD algorithm takes less time to execute than the original algorithm (DEA-R), then the following inequality holds for all the values of $j (1 \le j \le k-2)$. In addition to this, for at least one value of $j$, this will be an strict inequality:
\begin{multline*} cost\left(f(a_{j+2},a_{j+3}) -\left \lfloor \frac{a_j}{a_{j+1}} \right\rfloor f(a_{j+1},a_{j+2})\right)+cost \left(\left\lfloor \frac{a_j}{a_{j+1}} \right \rfloor\right)+
\\cost\left(\left \lfloor \frac{a_j}{a_{j+1}}\right \rfloor f(a_{j+1},a_{j+2})\right)\le cost \left(c-f(a_{j+1},a_{j+2})a_j\right)+ \\cost\left(\frac{c-f(a_{j+1},a_{j+2})a_j}{a_{j+1}}\right)+cost\left(f(a_{j+1},a_{j+2})a_j\right) \end{multline*}
Here the first cost on both sides of inequality is for subtraction operation followed by cost of division and then cost of multiplication respectively.
If we assume that the cost of addition (subtraction) and multiplication (division) of an  $m$ bit and $n$ bit integer are $m+n$ and $mn$ respectively then the above expression can be rewritten as follows:

\begin{multline*} \implies \log{\left|f(a_{j+2},a_{j+3})\right|}+ \log{\left|\left \lfloor \frac{a_j}{a_{j+1}} \right\rfloor f(a_{j+1},a_{j+2})\right|}+\log{\left|{a_j}\right|} \log{\left|{a_{j+1}} \right|}+
\\\log{\left|\left \lfloor \frac{a_j}{a_{j+1}}\right \rfloor\right|} \log{\left|f(a_{j+1},a_{j+2})\right|} \le  \log{\left|c\right|}+\log{\left|f(a_{j+1},a_{j+2})a_j\right|}+ \\\log{\left|{c-f(a_{j+1},a_{j+2})a_j}\right|}\log{\left|{a_{j+1}}\right|}+\log{\left|f(a_{j+1},a_{j+2})\right|}\log{\left|a_j\right|} \end{multline*}
The base of logarithms is $2$. Simplifying further, we get the following:

\begin{multline*} \implies \log{\left|f(a_{j+2},a_{j+3})\right|}+ \log{\left|\frac{1}{a_j}\left \lfloor \frac{a_j}{a_{j+1}} \right\rfloor \right|}+\log{\left|{a_j}\right|} \log{\left|{a_{j+1}} \right|}+
\\\log{\left|\left \lfloor \frac{a_j}{a_{j+1}}\right \rfloor\right|} \log{\left|f(a_{j+1},a_{j+2})\right|}-\log{\left|f(a_{j+1},a_{j+2})\right|}\log{\left|a_j\right|} \\ \le  \log{\left|c\right|}+\log{\left|{c-f(a_{j+1},a_{j+2})a_j}\right|}\log{\left|{a_{j+1}}\right|}\end{multline*}

\begin{multline*} \implies \log{\left|f(a_{j+2},a_{j+3})\right|}+ \log{\left|\frac{1}{a_j}\left \lfloor \frac{a_j}{a_{j+1}} \right\rfloor \right|}+\log{\left|{a_j}\right|} \log{\left|{a_{j+1}} \right|}+
\\\log{\left|f(a_{j+1},a_{j+2})\right|}\log{\left|\frac{1}{a_j}\left \lfloor \frac{a_j}{a_{j+1}} \right\rfloor\right|} \\ \le  \log{\left|c\right|}+\log{\left|{c-f(a_{j+1},a_{j+2})a_j}\right|}\log{\left|{a_{j+1}}\right|}\end{multline*}
On simplifying further, we get the following expression: 

\begin{multline*} 
\left|f(a_{j+2},a_{j+3})\left(\frac{1}{a_j}\left\lfloor \frac{a_j}{a_{j+1}}\right\rfloor\right)^{\left(1+\log{\left| f(a_{j+1},a_{j+2})\right|}\right)}\right| 
 \le \left|c\left(\frac{c}{a_j}-f(a_{j+1},a_{j+2})\right)^{\log\left|{a_{j+1}}\right|}\right|\end{multline*}

Assuming $f(a_{j+1},a_{j+2})=x$ and $\left\lfloor\frac{a_j}{a_{j+1}}\right\rfloor=y$

\begin{multline*}
\left| f(a_{j+2},a_{j+3}) \left(\frac{y}{a_j}\right)^{1+\log |x|}  \right| \le 
\left|c\left(\frac{c}{a_j}-x\right)^{\log{|a_{j+1}|}}\right|
\end{multline*}

\begin{multline*} \implies
\left| f(a_{j+2},a_{j+3}) \left(\frac{y}{a_j}\right)^{1+\log |x|}  \right| \le 
\left|c\left(\frac{c}{a_j}-\frac{c-f(a_{j+2},a_{j+3})a_{j+1}}{a_{j+2}}\right)^{\log{|a_{j+1}|}}\right|
\end{multline*}

\begin{equation}\label{cond15} \implies
\left| f(a_{j+2},a_{j+3}) \left(\frac{y}{a_j}\right)^{1+\log |x|}  \right| \le 
\left|c\left(\frac{ca_{j+2}-ca_j}{a_ja_{j+2}}+\frac{f(a_{j+2},a_{j+3})a_{j+1}}{a_{j+2}}\right)^{\log{|a_{j+1}|}}\right|
\end{equation}

\begin{equation}\label{cond16} \implies
\left( f(a_{j+2},a_{j+3}) \left(\frac{y}{a_j}\right)^{1+\log |x|}  \right)^2 \le 
\left(c\left(\frac{ca_{j+2}-ca_j}{a_ja_{j+2}}+\frac{f(a_{j+2},a_{j+3})a_{j+1}}{a_{j+2}}\right)^{\log{|a_{j+1}|}}\right)^2
\end{equation}

\begin{multline}\label{cond17} \implies
\left( f(a_{j+2},a_{j+3}) \left(\frac{y}{a_j}\right)^{1+\log |x|}  \right)^2 \le 
\\
c^2\left(\left(\frac{ca_{j+2}-ca_j}{a_ja_{j+2}}\right)^2+\left(\frac{f(a_{j+2},a_{j+3})a_{j+1}}{a_{j+2}}\right)^2
+2\left(\frac{ca_{j+2}-ca_j}{a_ja_{j+2}}\right)\left(\frac{f(a_{j+2},a_{j+3})a_{j+1}}{a_{j+2}}\right) \right)^{\log{|a_{j+1}|}}
\end{multline}

Note that $\frac{y}{a_j}$ is less than 1
and since $a_{j+1}> a_{j+2}$,

\begin{equation} \label{ineq17}
\left(f(a_{j+2},a_{j+3})\right)^2 \left(\frac{y}{a_j}\right)^{2+2\log |x|} < \left(\frac{f(a_{j+2},a_{j+3})a_{j+1}}{a_{j+2}}\right)^2    
\end{equation}

We know that $a_{j+2}<a_j$, so $\left(\frac{ca_{j+2}-ca_j}{a_ja_{j+2}}\right)<0$ if and only if $c>0$. Since $a_{j+1}$ and $a_{j+2}$ both are positive,  
   $\left(\frac{f(a_{j+2},a_{j+3})a_{j+1}}{a_{j+2}}\right)<0$ if and only if  $f(a_{j+2},a_{j+3})<0$. 
   \newline Hence  $2\left(\frac{ca_{j+2}-ca_j}{a_ja_{j+2}}\right)\left(\frac{ca_{j+2}-ca_j}{a_ja_{j+2}}\right)>0$ implies one of the following:
   \begin{enumerate}
       \item $c>0$ and $f(a_{j+2},a_{j+3})<0$

      \item $c<0$ and $f(a_{j+2},a_{j+3})>0$, 
   \end{enumerate}
Thus in any of the above two conditions, inequality \ref{cond17} is satisfied. For the other two conditions, we need to further analyse values of $c$ and $f$, so that 

$$\left(\frac{ca_{j+2}-ca_j}{a_ja_{j+2}}\right)^2+2\left(\frac{ca_{j+2}-ca_j}{a_ja_{j+2}}\right)\left(\frac{f(a_{j+2},a_{j+3})a_{j+1}}{a_{j+2}}\right) >0$$
We can control the value of $c$ but not the value of function $f$. So, the theoretical analysis does not conclude that by increasing or decreasing the value of $c$, relation \ref{cond17} will necessarily be satisfied. It also does not conclude anything about which  values of $c,a$ or $b$ satisfy inequality \ref{cond17}.

\end{document}